\documentclass{article}

\usepackage{amssymb,amsmath,amsthm}
\usepackage[english]{babel}
\usepackage{algorithm}
\usepackage{algpseudocode}

\newcommand{\Nat}{{\mathbb N}}
\newcommand{\Real}{{\mathbb R}}
\newcommand{\st}{\;|\;}
\newcommand{\set}[1][ ]{\{ #1 \}}

\newcommand{\R}{\mathcal{R}}
\newcommand{\tendingtoinfty}[1]{\stackrel{#1 \rightarrow \infty}{\longrightarrow}}

\newtheorem{lem}{Lemma}
\newtheorem{thm}[lem]{Theorem}

\begin{document}

\title{Revisiting Deadlock Prevention: A Probabilistic Approach}

\author{Fabiano de S. Oliveira\\
Valmir C. Barbosa\thanks{Corresponding author (valmir@cos.ufrj.br).}\\
\\
Programa de Engenharia de Sistemas e Computa\c c\~ao, COPPE\\
Universidade Federal do Rio de Janeiro\\
Caixa Postal 68511\\
21941-972 Rio de Janeiro - RJ, Brazil}

\date{}

\maketitle

\begin{abstract}
We revisit the deadlock-prevention problem by focusing on priority digraphs
instead of the traditional wait-for digraphs. This has allowed us to formulate
deadlock prevention in terms of prohibiting the occurrence of directed cycles
even in the most general of wait models (the so-called AND-OR model, in which
prohibiting wait-for directed cycles is generally overly restrictive). For a
particular case in which the priority digraphs are somewhat simplified, we
introduce a Las Vegas probabilistic mechanism for resource granting and analyze
its key aspects in detail.

\bigskip
\noindent
\textbf{Keywords:}
Deadlock prevention, Priority digraphs, Probabilistic algorithms.
\end{abstract}

\section{Introduction}

In any computation, processes need resources in order to carry out their tasks.
A \emph{resource} is either a physical device (such as a printer, a CPU, a hard
disk, network bandwidth, etc.) or a logical device (such as a TCP port, the
position of an array in a data buffer, etc.). In general, resources are
expensive and, therefore, they exist in limited amounts. They must be shared by
the processes, which in turn must use them in such a way that no conflict arises
due to concurrent access. Informally, processes must not compute on shared
resources until it becomes safe to do so, which normally is taken to hold true
when they are \emph{granted} the resources they need. When such a wait turns out
to be indefinite, we say that the computation is in a \emph{deadlock} state. 

Since deadlock characterizations depend on how the waits among processes occur,
they have been based on some assumed wait model
\cite{MisraChandy, BrachaToueg, KshemkalyaniSinghal, RyangPark, BrzezinskiHelaryRaynalSinghal, Barbosa99thecombinatorics},
the most general one being the so-called AND-OR model. Such a model in essence
allows unconstrained waits to take place. It arises particularly in the scenario
where the computation requires several non-singleton groups of resources, each
having equivalent resource instances, and the processes request several
resources at once. We henceforth assume that this is the model under which waits
occur in the computations we consider, and note that the necessary and
sufficient condition for deadlocks to arise in this case is that the underlying
wait-for digraph contain a so-called b-knot
(cf.~\cite{Barbosa99thecombinatorics} and references therein). 

Preventing deadlocks in this context while assuming the usual necessary
condition that do not directly bear on the structure of the wait-for digraph is
then seen to require that, by design of the resource-granting mechanism, b-knots
never occur. This, however, seems unfeasible, which incidentally is why
prevention approaches have shunned the real problem and relied instead on
forbidding the occurrence of directed cycles in the waif-for digraph. These, of
course, are themselves necessary for b-knots to exist, but in general constitute
a much more restrictive necessary condition and prohibiting their appearance is
bound to rule out several deadlock-free ways in which the computation might
unfold. 

Here we revisit deadlock prevention by first diverting the focus away from
wait-for digraphs. We focus instead on what we call priority digraphs, thereby
putting aside the complications associated with b-knots and replacing them with
the more natural interactions of directed cycles involving priorities. While
this gives us a better conceptual handle on the problem, for computations
comprising a large number of processes there remains little hope of efficient
prevention. What we do then is to introduce a Las Vegas probabilistic mechanism
for request granting and analyze its deadlock-related properties for some cases
of interest.

Some particular graph-theoretic notations are required for what follows. A graph
$G$ has a set $V(G)$ of vertices and a set $E(G)$ of edges where each edge is a
distinct pair of vertices. When the pair of vertices in each edge is ordered, we
say that such a graph is a digraph, its edges are arcs, and denote the set of
arcs by $A(G)$. The degree $d(v)$ of a vertex $v$ in a graph $G$ is the
cardinality of the set $\set[w \in V(G) \;|\; vw \in E(G)]$. Finally, the
maximum degree in a graph $G$ is
$\Delta(G) = \max \set[ d(v) \;|\; v \in V(G) ]$. 

\section{Definitions}

We assume that the computation under study takes place in the fully asynchronous
distributed model of~\cite{LivroValmir}. In such a model, each process computes
on an independent clock and the communication among processes is effected
through the exchange of messages. Each message is sent in a point-to-point
fashion and delivered in a finite amount of time, although the exact delay is
not predictable. Messages are delivered on logical bidirectional channels that
exist between any two processes that need to communicate with each other.

A \emph{resource class} is a set of resources sharing the same properties and
providing services in such a way that any two resources from this set are
considered equivalent to each other by the processes. Therefore, a \emph{disk}
can be considered a resource class, whereas the distinct hard disks named
\emph{hd1}, \emph{hd2}, and \emph{hd3} are members of the resource class
\emph{disk}. However, a printer named \emph{prt1} will usually be a member of a
distinct resource class (say, \emph{printer}) due to the differences in
properties between disks and printers and the services they provide (although
they could be members of the same resource class \emph{output device} in some
particular application). 

Consider a computation in progress at instant $t \in \Real^+$. Each process of
such a computation is identified by a distinct natural number. Let
$\mathcal{P}^t \subseteq \Nat$ be the set of identifiers of the processes which
are requesting and/or holding resources at instant $t$. Equivalently,
$\mathcal{P}^t$ refers to processes waiting for others to release some required
resources and/or being waited for to release resources that they hold at instant
$t$. The process identified by the number $i \in \mathcal{P}^t$ will be denoted
by $P_i$. Moreover, each resource class is identified by a distinct natural
number and we denote by $\R^t \subseteq \Nat$ the set of identifiers of the
resource classes from which there exist resources requested or held at instant
$t$. The resource class identified by the number $r \in \R^t$ will be denoted by
$\R^t_r$. Finally, each resource is identified by a distinct natural number,
denoting by $R_r$ the resource identified by the number $r \in \Nat$. The
resources that matter for such a computation at instant $t$ are then seen to be
members of the set $\bigcup_{r \in \R^t} \R^t_r$. The \emph{computation graph}
$G^t$ is the graph such that $V(G^t) = \mathcal{P}^t$ and $ij \in E(G^t)$ if and
only if both $P_i$ and $P_j$ request and/or hold resources from $\R^t_r$, for
some $r \in \R^t$. The \emph{wait-for digraph} $W^t$ of such a computation at
instant $t$ models the waits among processes at this instant, that is,
$V(W^t) = \mathcal{P}^t$ and $ij \in A(W^t)$ if and only if $P_j$ holds
resources from a resource class from which resources are needed by $P_i$. When
the instant $t$ is clear from the context or unimportant in it, we simply omit
it from all notations. 

According to our definition, $\mathcal{P}^t$ (and therefore $G^t$, $W^t$,
$\R^t$, etc.) can be significantly large. Hence, we assume that all this
information is stored distributively, and also that the periodic recording of
snapshots is unfeasible. Still regarding cardinalities, note that $\R^t_r$ is
always finite for any $r \in \R^t$: an infinite $\R^t_r$ can be ignored since
its resources will never be decisively involved in a deadlock situation.
Besides, the set of resources being requested by a single process is always
finite as well, or else the process would spend an infinite amount of time to
request them. Therefore, each vertex of $W^t$ has a finite out-degree, whereas
the in-degree may be infinite.

\section{The computations that we consider}

A process, in principle, may request and/or release resources at any instant
during its execution. However, as we discuss later on in this section, allowing
such a fully unconditional behavior is problematic from the standpoint of
preventing deadlocks. We assume a more restrictive protocol to be followed by
the distributed algorithms regarding the requesting and releasing of resources.
This protocol is presented by means of a template for the distributed
algorithms. 

For all $i \in \mathcal{P}^t$, $r \in \R^t$, let $x^{i,t}_r$ be the number of
resources from $\R^t_r$ that $P_i$ requests at instant $t \in \Real^+$. We
present a template in Algorithm~\ref{Template} for specifying the local
processing by $P_i$. 

\begin{algorithm}[t]
	\caption{Algorithm for $P_i$.}
	\label{Template}
	\begin{algorithmic}[1]
		\Procedure{Template}{} 
			\While{(there remains work to be done)}		
				\State \parbox[t]{0.8\textwidth}{Compute without using shared resources.}
				\State \parbox[t]{0.8\textwidth}{\textbf{Request phase}: \parbox[t]{0.6\textwidth}{For each $r \in \R$, request $x^i_r$ resources from $\R_r$. Wait until all of them are granted.}}
				\State \parbox[t]{0.8\textwidth}{Compute using the granted resources.}
				\State \parbox[t]{0.8\textwidth}{\textbf{Release phase}: Release all granted resources.} \label{Template-Release}
			\EndWhile
		\EndProcedure
	\end{algorithmic}
\end{algorithm}

While we can offer no indisputable argument that all computations can be cast
into the template of Algorithm~\ref{Template}, our intention in focusing only
on those that can is to allow any resource-granting mechanism one may come up
with to work only on a full set of requests for each participating process. Were
it not so, then any such mechanism would have to deal with a host of the
classical deadlock-inducing pitfalls, such as the one in which two processes
already hold each a resource, from distinct singleton resource classes, and they
engage in requesting each the resource that the other holds.

Notice also that Algorithm~\ref{Template} does not per se negate any of the
necessary conditions for deadlocks to occur (not even the hold-and-wait
condition, which is precisely what happens, in general, in the request phase).
Our problem henceforth is to provide deadlock prevention through a
request-granting mechanism (RGM) interacting with the request and release phases
of all processes involved in the computation. Our solution will strive to do so
while providing as much concurrency among processes as possible.

\section{Prevention in the AND-OR wait model}

Consider a computation at instant $t \in \Real^+$ and let $i \in \mathcal{P}^t$.
For each $r \in \R^t$, we denote by $g^{i,t}_r$ the number of resources from
$\R^t_r$ which are currently granted to $P_i$. Clearly,
$g^{i,t}_r \leq x^{i,t}_r \leq |\R^t_r|$. Denote by $\mathrm{Max}(P)$ the set of
maximum elements of an order $P = (X, \prec)$, i.e.,
$\mathrm{Max}(P) = \set[ i \in X \st \not \exists$ $j \in X$ such that
$i \prec j$ $]$.

We describe a characterization for a deadlock-free RGM as follows. For
$\mathcal{P}^t_r = \set [i \in \mathcal{P}^t \st x^{i,t}_r > 0]$, we say that an
RGM is \emph{driven} by a family of orders
$\mathcal{O}^t = \set [(\mathcal{P}^t_r,$ $\prec^t_r) \st$ $r \in\mathcal{R}^t]$
if, for each $r \in \mathcal{R}^t$, two conditions hold: 

\begin{enumerate}
\item Only processes in $\mathrm{Max}(\mathcal{P}^t_r)$ are granted resources; 
\item Any order $P^t_r$ obtained from the repeated removal of a maximum element
from $(\mathcal{P}^t_r, \prec^t_r)$ is such that
$\displaystyle\sum_{i \in \mathrm{Max}(P^t_r)} x^{i,t}_r\leq |\mathcal{R}^t_r|$.
\end{enumerate}

The \emph{priority digraph} $D(\mathcal{O}^t)$ is the digraph
$(\mathcal{P}^t, A^t)$ such that $ij \in A^t$ if and only if, for some
$r \in \mathcal{R}^t$, $i \prec^t_r j$ and there does not exist
$z \in \mathcal{P}^t_r$ such that $i \prec^t_r z \prec^t_r j$.

\begin{thm}
\label{CaracLivreDeadlock:Teorico}
An RGM is deadlock-free if and only if it is driven by a family of orders
$\mathcal{O}^t = \set [(\mathcal{P}^t_r, \prec^t_r) \st r \in \mathcal{R}^t]$
such that $D(\mathcal{O}^t)$ is acyclic for each instant $t \in \Real^+$.
\end{thm}

\begin{proof}
Let $\mathcal{A}$ be a deadlock-free RGM. At each instant $t \in \Real^+$, the
family of orders
$\mathcal{O}^t = \set [(\mathcal{P}^t_r, \prec^t_r) \st r \in \mathcal{R}^t]$ is
that defined as follows: $i \prec^t_r j$ precisely when part of the resources
from $\R^t_r$ used by $P_i$ during the computation was previously used by $P_j$.
Because $\mathcal{A}$ is deadlock-free, it can be seen to be driven by
$\mathcal{O}^t$ due to the following facts: (i) the sequence of resource
releases by processes corresponds to the operation of repeatedly removing
maximum elements from $(\mathcal{P}^t_r, \prec^t_r)$ for each
$r \in \mathcal{R}^t$; and (ii) the number of available resources does not
exceed the number of granted resources at any time. Now, suppose that
$D(\mathcal{O}^t)$ is cyclic. Therefore, there exist $r_1, \ldots, r_L$ and
$a_1, \ldots, a_L$ such that
$a_1 \prec^t_{r_1} a_2 \prec^t_{r_2} a_3 \prec^t_{r_3} \cdots \prec^t_{r_{L-1}} a_L \prec^t_{r_L} a_1$.
Note that if $i \prec^t_r j$ for some $r \in \mathcal{R}^t$, then the instant at
which $P_i$ is granted all its required resources is greater than the instant at
which $P_j$ releases all its resources. By transitivity on the assumed cycle, a
contradiction is easily obtained.

Conversely, let $\mathcal{A}$ be an RGM and, for each $t \in \Real^+$, let
$\mathcal{O}^t = \set [(\mathcal{P}^t_r, \prec^t_r) \st r \in \mathcal{R}^t]$
be a family of orders such that $\mathcal{A}$ is driven by $\mathcal{O}^t$ with
$D(\mathcal{O}^t)$ acyclic. For any particular $T \in \Real^+$, we show that
each process will be granted its requested resources in a finite amount of time
by induction on $|\mathcal{P}^T|$, therefore showing that $\mathcal{A}$ is
deadlock-free. If $|\mathcal{P}^T| = 1$, then the claim is trivial. Suppose
$|\mathcal{P}^T| > 1$ and that the claim holds for all orders whose process sets
are strictly contained in $\mathcal{P}^T$. Since $D(\mathcal{O}^T)$ is acyclic,
there exists a vertex $i$ which is a sink of $D(\mathcal{O}^T)$. From condition
$2$ for an RGM to be said to be driven by a family of orders, it follows that it
is possible to grant resources to all processes in
$\mathrm{Max}(\mathcal{P}^T_r)$ for each $r \in \mathcal{R}^T$. By condition
$1$, $\mathcal{A}$ eventually grants $P_i$ all its required resources in finite
time. In finite time, $P_i$ will release all such resources. Clearly, by the
induction hypothesis, all processes in $\mathcal{P}^T \setminus P_i$ are granted
their resources in finite time. Consequently, the same holds for
$\mathcal{P}^T$.
\end{proof}

To illustrate the use of the priority-based approach to which
Theorem~\ref{CaracLivreDeadlock:Teorico} refers, we now consider the classical
prevention strategy that forces processes to follow a pre-established linear
order of resource classes inside each of the request phases of
Algorithm~\ref{Template}. In general, singleton resource classes are assumed in
such a strategy. Casting it into our priority-based terms requires not only the
linear order, which we let $(\bigcup_{t \geq 0} \R^t, \prec)$ be, but also that
we specify the driving family of orders
$\mathcal{O}^t = \set [(\mathcal{P}^t_r, \prec^t_r) \st r \in \mathcal{R}^t]$ at
each instant $t \in \Real^+$ for the corresponding RGM, which we denote by
$\mathcal{C}$. We do this by letting
\begin{center}
	$C^t_r = \set[ i \in \mathcal{P}^t_r \st g^{i,t}_r < x^{i,t}_r$ and $\forall r \prec s$, $g^{i,t}_s = x^{i,t}_s ]$,
\end{center}
for all $t \in \Real^+$ and $r \in \mathcal{R}^t$, and then letting
$\prec^t_r$ be a linear order of $\mathcal{P}^t_r$ such that for all
$i, j \in \mathcal{P}^t_r$,
\begin{center}
	$(i \in C^t_{r_1}$ and $j \in C^t_{r_2}$ with $r_2 \prec r_1) \Longrightarrow i \prec^t_r j$.
\end{center}

Since $\mathcal{C}$ is by construction driven by $\mathcal{O}^t$ for each
$t \in \Real^+$, by Theorem~\ref{CaracLivreDeadlock:Teorico} it suffices to
prove that $D(\mathcal{O}^t)$ is acyclic in order for $\mathcal{C}$ to be
deadlock-free. This can be done as follows. Suppose there exists a cycle in
$D(\mathcal{O}^t)$ for some $t \in \Real^+$. Therefore, there exist
$p_1,\ldots,p_L \in \mathcal{P}^t$ and $r_1,\ldots,r_L \in \R^t$, $L \geq 2$,
such that
$p_1 \prec^t_{r_1} p_2 \prec^t_{r_2} \cdots \prec^t_{r_{L-1}} p_L \prec^t_{r_L} p_1$.
Since $p_1 \prec^t_{r_1} p_2$, then either $p_1$ and $p_2$ belong to a same set
$C^t_z$ or they belong, respectively, to $C^t_{z_1}$ and $C^t_{z_2}$ with
$z_2 \prec z_1$. By transitivity, it follows that $p_1, \ldots, p_L \in C^t_z$,
and therefore they should be in linear order, which is a contradiction. 

The same driving family of orders defined above can be used for the classical
prevention strategy with non-singleton resource classes, although it will not
yield the best possible concurrency. The driving family of orders can be easily
changed in this case to improve concurrency by letting $\prec^t_r$ be a partial
order instead of a linear order.

\section{A basic deadlock-free Las Vegas RGM}

Consider a computation at any given instant $t \in \Real^+$. We assume that no
single computer has enough memory or processing capacity to store and process
all the data relating to the various structures we have seen so far. In this
section, we present a deadlock-free RGM which is both simple and fully
distributed, being therefore amenable to deployment in significantly large
(formally infinite) systems. 

Note that each $\R_r$, $r \in \R$, corresponds to a clique (not necessarily
maximal) of $G$. The converse does not hold in general. Therefore, the size of
$\R$ is limited by the number of edges of $G$. Since the complexity of an RGM in
general increases as $|\R|$ increases, for a worst-case analysis we assume from
now on that each edge of $G$ represents a resource class containing exactly one
resource. An example application which fits this assumption naturally is that of
the communication channels between the processes. Although bidirectional, when
half-duplex they cannot transmit in both directions simultaneously. Therefore,
each communication channel is a resource accessed by the pair of processes which
it links together. By Theorem~\ref{CaracLivreDeadlock:Teorico}, the goal of a
deadlock-free RGM in this case is to orient the edges of $G$ acyclically.

Our approach to obtain such an acyclic orientation is based on finding
independent sets of $G$. The idea is to grant the resources to processes in
independent sets. In such an approach, the larger each independent set, the more
concurrency the computation is expected to be able to achieve. Note that the
number of processes using shared resources at any instant is clearly bounded
from above by $\omega(G)$, the size of $G$'s largest independent set. 

It is well-known that determining $\omega(G)$ for a general graph $G$ is an
NP-hard problem \cite{GareyJohnson}. Our strategy for generating independent
sets of $G$ will be to generate digraphs $D$ obtained from $G$ by orienting its
edges and using the set of sinks of $D$ as the independent set. In fact, not
only is the set of sinks of a digraph $D$ an independent set of $G$, but also
each maximal independent set of $G$ is the set of sinks for some digraph $D$ of
$G$. The problem is reduced therefore to finding ``good'' digraphs $D$, that is,
digraphs which maximize the cardinality of the set of sinks. The general
approach is given in Algorithm~\ref{RGMOrientRandom}. Although the algorithm is
specified in a centralized manner, note that implementing it in a distributed
fashion with $O(1)$ time complexity is straightforward. Each process needs only
to agree on the orientation of each incident edge with the corresponding
neighbor and a process accesses the requested resources precisely when it
becomes a sink. Each new iteration may be implemented as each process initiating
the negotiation of a new orientation of its incoming edges with its neighbors.

\begin{algorithm}[t]
	\caption{Random-Orientation RGM (centralized version).}
	\label{RGMOrientRandom}
	\begin{algorithmic}[1]
		\Procedure{Random-Orientation-RGM}{} 
			\While{$\mathcal{P}^t$ is nonempty}
				\State $D$ $\gets$ random orientation of $G^t$, both directions of each edge having the same probability.
				\State Grant the sinks of $D$ all resources they require.
			\EndWhile
		\EndProcedure
	\end{algorithmic}
\end{algorithm}

Assume that the digraph $D$ is obtained by sequentially orienting each edge of
$G$ randomly such that both orientations for each edge are equally likely to
occur. Note that this strategy does not take into account the orientations done
so far at any given instant, and therefore may be subject to improvements by
using this piece of information in order to maximize the expected number of
sinks. Nevertheless, it is the RGM of interest in this paper and we study its
effectiveness. In order to do so, we define the random variable $X_n$ to be the
number of sinks of $D$, for $n$ the number of vertices of $D$, and analyze two
quantities: the probability $\mathrm{Pr}[X_n > 0]$ of generating at least one
sink and the expected number $\mathrm{E}[X_n]$ of sinks. The former indicates
how likely it is for the computation to make progress at each stage. The latter
is used to derive the expected time of such a Las Vegas algorithm. In fact, if
$T(n)$ is the random variable corresponding to the time complexity of the
algorithm on $n$ initial processes when no new processes can be awaken during
the execution, then $\mathrm{E}[T(n)]$ is defined recursively by
$\mathrm{E}[T(n)] = 1 + \mathrm{E}[T(n - \mathrm{E}[X_n])]$.

Clearly, $\Pr[X_n < 0] = \Pr[X_n > \omega(G)] = 0$. In the following
subsections, we derive expressions for $\Pr[X_n > 0]$ and $\mathrm{E}[X_n]$
restricted to distinct classes of graphs. The final subsection summarizes the
results.

While working out $\mathrm{E}[X_n]$, we often use the auxiliary Bernoulli random
variable $Y_v$, $v \in V(G)$, defined as $Y_v = 1$ if the vertex $v$ is a sink
in $D$, $Y_v = 0$ otherwise. Therefore, $X_n = \sum_{v \in V(G)} Y_v$ and
consequently
$\mathrm{E}[X_n] = \mathrm{E}[\sum_{v \in V(G)} Y_v] = \sum_{v \in V(G)} \mathrm{E}[Y_v] = \sum_{v \in V(G)} \Pr[Y_v = 1]$.

\subsection{Trees}

Let $T_n$ denote a general tree on $n$ vertices. Trees are connected graphs free
of cycles by definition, and thus no cycles can be formed in $D$ either.
Therefore, trivially $\Pr[X_n > 0] = 1$. We work out $\mathrm{E}[X_n]$ for two
subclasses of trees, stars $S_n$ on $n$ vertices and paths $P_n$ on $n$
vertices, which have respectively the largest and the smallest value for
$\omega(T_n)$. In fact, the following lemma describes the bounds on the size of
an independent set of $T_n$, for $n \geq 2$. Trivially, $\omega(T_n) = 1$ for
$n = 1$.

\begin{lem}
\label{WTnWPn}
$\left\lceil \frac{n}{2} \right\rceil \leq \omega(T_n) \leq n - 1$, for each
$n \geq 2$.
\end{lem}

\begin{proof}
Obviously, $\omega(T_n) \leq n - 1$. For stars, $\omega(S_n) = n - 1$. On the
other hand, we prove that
$\omega(T_n) \geq \omega(P_n) = \left\lceil n/2 \right\rceil$, thus establishing
the result.

Clearly, the result holds for $n = 2$. Suppose it holds for all trees having
fewer than $n \geq 3$ vertices. Let $u$ be a leaf of $T_n$ and let
$uv \in E(T_n)$. Let $T'$ be the forest obtained by deleting vertices $u$ and
$v$ from $T_n$, and let $T'_1,\ldots,T'_k$ be the connected components of $T'$.
Let $n'_i = |V(T'_i)|$ for each $1 \leq i \leq k$. By the induction hypothesis,
$\omega(T'_i) \geq \omega(P_{n'_i})$ for each $1 \leq i \leq k$ and therefore
$\sum_{i=1}^k \omega(T'_i) \geq \sum_{i=1}^k \omega(P_{n'_i})$. Clearly,
$\omega(T_n) - 1 \geq \omega(T') = \sum_{i=1}^k \omega(T'_i) \geq \sum_{i=1}^k \omega(P_{n'_i}) \geq \omega(P_{n-2}) = \omega(P_n) - 1$.
\end{proof}

\begin{thm}
\label{EXSn}
If $G$ is a star $S_n$, then
$\mathrm{E}[X_n] = \frac{n-1}{2} + \frac{1}{2^{n-1}}$, for each $n \geq 1$.
\end{thm}

\begin{proof}
Let $u$ be the universal vertex of $S_n$. Therefore,
\begin{eqnarray*}
	 \mathrm{E}[X_n] & = & \Pr[Y_u = 1] + \sum_{v \in V(G)\setminus u} \Pr[Y_v = 1] \\
	      & = & \frac{1}{2^{n-1}} + (n-1) \times \frac{1}{2} \\
		  & = & \frac{n-1}{2} + \frac{1}{2^{n-1}} \;{.}
\end{eqnarray*}
\end{proof}

\begin{thm}
\label{EXPn}
If $G$ is a path $P_n$, then $\mathrm{E}[X_n] = \frac{n+2}{4}$, for each
$n \geq 1$.
\end{thm}

\begin{proof}
Let $u,w$ be the vertices of $P_n$ having unit degree. Therefore,
\begin{eqnarray*}
	 \mathrm{E}[X_n] & = & \Pr[Y_u = 1] + \Pr[Y_w = 1] + \sum_{v \in V(G)\setminus \set [u,w]} \Pr[Y_v = 1] \\
	    & = & 2 \times \frac{1}{2} + (n-2) \times \frac{1}{4}  \\
		& = & \frac{n+2}{4} \;{.}
\end{eqnarray*}
\end{proof}

\subsection{Cycles}

The simplest class of graphs containing at least one cycle is that of the cycles
itself. Denote by $C_n$ a cycle having $n$ vertices. The following theorem is
straightforward.

\begin{thm}
\label{PXEXCn}
If $G$ is a cycle $C_n$, then $\Pr[X_n > 0] = 1 - \frac{1}{2^{n-1}}$ and
$\mathrm{E}[X_n] = \frac{n}{4}$, for each $n \geq 3$.
\end{thm}

\begin{proof}
Clearly, 
\begin{eqnarray*}
	 \Pr[X_n > 0] & = & 1 - \Pr[X_n = 0] \\
				  & = & 1 - 2 \times \frac{1}{2^n} \\
				  & = & 1 - \frac{1}{2^{n-1}}
\end{eqnarray*}
and
\begin{eqnarray*}
	 \mathrm{E}[X_n] & = & \sum_{v \in V(G)} \Pr[Y_v = 1] \\
				     & = & n \times \frac{1}{4} \\
				     & = & \frac{n}{4} \;{.}
\end{eqnarray*}
\end{proof}

Note that, in particular, $\omega(C_n) = \lfloor n/2 \rfloor$.

\subsection{Complete graphs}

Consider the class of graphs containing the maximum possible number of cycles in
a graph, that is, graphs for which any subset of vertices induces a cycle. This
class is that of the complete graphs. Denote by $K_n$ a complete graph having
$n$ vertices. The odds of obtaining a sink in a random orientation of such a
graph are the worst possible, since $\omega(G) = 1$ in this case.

\begin{thm}
\label{PXEXKn}
If $G$ is a complete graph $K_n$, then
$\Pr[X_n > 0] = \mathrm{E}[X_n] = \frac{n}{2^{n-1}}$, for each $n \geq 1$.
\end{thm}

\begin{proof}
Clearly, 
\begin{eqnarray*}
	 \Pr[X_n > 0] & = & \Pr[X_n = 1] \\
				  & = & \Pr[\bigvee_{v \in V(G)} Y_v = 1] \\
			      	& = & \sum_{v \in V(G)} \Pr[Y_v = 1] \\
				  & = & n \times \frac{1}{2^{n-1}} \\
				  & = & \frac{n}{2^{n-1}} \;{.}
\end{eqnarray*} 
Also, the above summation yields $\mathrm{E}[X_n]$ as well.
\end{proof}

\subsection{Bounded-degree graphs}

The previous graph classes (trees, cycles, and complete graphs) were considered
based on their numbers of cycles (respectively, no cycles, exactly one cycle,
and the maximum possible number of cycles). We now consider graphs often claimed
to arise from practical applications. In the particular context of the present
section, recall that each edge of $G$ corresponds to a resource class. Since
$\R(i) = \set [r \in \R \st x^i_r > 0]$ is finite for each $i \in \mathcal{P}$,
therefore $G$ is a bounded-degree graph. More precisely,
$\Delta(G) = \max \set[ |\R(i)| \st i \in \mathcal{P} ]$ is a constant. And even
though this need not hold for $G$ in general, we now proceed to calculate
$\Pr[X_n > 0]$ and $\mathrm{E}[X_n]$ when degrees are bounded. Let $B_{n, k}$
denote a graph $G$ on $n$ vertices such that $\Delta(G) = k$.

\begin{thm}
\label{PXBn}
If $G$ is a bounded-degree graph $B_{n, k}$, then
$\Pr[X_n > 0] \geq 1 - \left( 1 - \frac{1}{2^k} \right)^{\left\lceil \frac{n}{k+1} \right\rceil }$,
for each $n \geq 1$, $k \geq 1$.
\end{thm}

\begin{proof}
Let $I$ be a maximum independent set of $G$, that is,
$|I| = \omega(G) \geq \left\lceil \frac{n}{k+1} \right\rceil$. Therefore,
\begin{eqnarray*}
	 \Pr[X_n = 0] & = & \Pr[\bigwedge_{v \in V(G)} Y_v = 0] \\
				  & \leq & \Pr[\bigwedge_{v \in I} Y_v = 0] \\
				  & = & \prod_{v \in I} \Pr[Y_v = 0] \\
	              & = & \prod_{v \in I} (1 - \Pr[Y_v = 1]) \\
                  & = & \prod_{v \in I} (1 - \frac{1}{2^{d(v)}}) \\
                  & \leq & \left( 1 - \frac{1}{2^k} \right)^{|I|} \\
	              & \leq & \left( 1 - \frac{1}{2^k} \right)^{\left\lceil \frac{n}{k+1} \right\rceil} \;{.}
\end{eqnarray*}
Consequently,
$\Pr[X_n > 0] = 1 - \Pr[X_n = 0] \geq 1 - \left( 1 - \frac{1}{2^k} \right)^{\left\lceil \frac{n}{k+1} \right\rceil }$.
\end{proof}

\begin{thm}
\label{EXBn}
If $G$ is a bounded-degree $B_{n, k}$, then
$\mathrm{E}[X_n] \geq \frac{n}{2^k}$, for each $n \geq 1$, $k \geq 1$.
\end{thm}
\begin{proof}
Clearly, 
\begin{eqnarray*}
	 \mathrm{E}[X_n] & = & \sum_{v \in V(G)} \Pr[Y_v = 1] \\
				     & = & \sum_{v \in V(G)} \frac{1}{2^{d(v)}} \\
					 & \geq & n \times \frac{1}{2^k} \\
					 & = & \frac{n}{2^k} \;{.} 
\end{eqnarray*}
\end{proof}

Note that, for $B_{n, k}$, $\Pr[X_n > 0] \tendingtoinfty{n} 1$ and
$\mathrm{E}[X_n] \tendingtoinfty{n} \infty$.

\subsection{Random graphs}

In this section we let $G$ be a random graph on $n$ vertices such that for each
distinct pair $u, v \in V(G)$, edge $uv$ is likely to exist with probability
$p$. This is the Erd\H{o}s-R\'{e}nyi random graph, classically denoted by
$G_{n,p}$. Clearly, in $G_{n,p}$ the probability that a randomly chosen vertex
$v$ has degree $d \geq 0$ is:
\begin{eqnarray*}
	 \Pr[d(v) = d] & = & \binom{n-1}{d} p^d (1-p)^{n-1-d} \\
			       & = & \binom{n-1}{d} \left( \frac{p}{1-p} \right)^d (1-p)^{n-1} \;{.}
\end{eqnarray*}
Letting $z = (n-1)p$ be the \emph{mean degree} of $G_{n,p}$, we obtain
\begin{eqnarray*}
	 \Pr[d(v) = d] & = & \binom{n-1}{d} \left( \frac{z}{n-1-z} \right)^d \left(1-\frac{z}{n-1}\right)^{n-1} \;{.}
\end{eqnarray*}
Clearly, 
\begin{eqnarray*}
	 \Pr[d(v) = d] \tendingtoinfty{n} & \frac{z^d}{d! e^z} \;{,} 
\end{eqnarray*}
which is the well-known Poison distribution.
	
\begin{thm}
\label{EXGnp}
If $G$ is a random graph $G_{n, p}$, then $\mathrm{E}[X_n] \approx \frac{n}{e^{z/2}}$, for each $n \geq 1$, $0 \leq p \leq 1$.
\end{thm}

\begin{proof}
We have 
\begin{eqnarray*}
	 \mathrm{E}[X_n] & = & \sum_{v \in V(G)} \Pr[Y_v = 1] \\
					& = & n \Pr[Y_1 = 1] \\
		  & = & n \sum_{d = 0}^{n - 1} \Pr[Y_1 = 1 \st d(1) = d] \times \Pr[d(1) = d] \\
		  & = & n \sum_{d = 0}^{n - 1} \frac{1}{2^d} \times \Pr[d(1) = d] \;{.}
\end{eqnarray*}
On the other hand,
\begin{eqnarray*}
	\sum_{d = 0}^{n - 1} \frac{1}{2^d} \times \Pr[d(1) = d] & \approx & \sum_{d \geq 0} \frac{1}{2^d} \times \Pr[d(1) = d] \\
    & = & \sum_{d \geq 0} \frac{1}{2^d} \times \frac{z^d}{d! e^z} \;{,}
\end{eqnarray*}
and thus,
\begin{eqnarray*}
	 \mathrm{E}[X_n] & \approx & n \sum_{d \geq 0} \frac{1}{2^d} \times \frac{z^d}{d! e^z} \\
					 & = & \frac{n}{e^z} \sum_{d \geq 0} \frac{(z/2)^d}{d!} \\
					 & = & \frac{n}{e^z} \times e^{z/2} \\
				     & = & \frac{n}{e^{z/2}} \;{.}
\end{eqnarray*} 
\end{proof}

\begin{thm}
\label{PXGnp}
If $G$ is a random graph $G_{n, p}$, then $\Pr[X_n > 0] \tendingtoinfty{n} 0$.
\end{thm}

\begin{proof}
By Markov's inequality,
$\Pr[X_n \geq 1] \leq \mathrm{E}[X_n] = n/e^{z/2} \tendingtoinfty{n} 0$.
\end{proof}

We remark that both Theorems~\ref{EXGnp} and~\ref{PXGnp} are given for a fixed
value of $p$. If, instead, it is $z$ that is fixed (i.e., $p$ is made
proportionally smaller as $n$ increases), then
$\mathrm{E}[X_n] \tendingtoinfty{n} \infty$. In this case, Theorem~\ref{EXGnp}
remains valid but Theorem~\ref{PXGnp} loses its meaning.

\subsection{Power-law random graphs}

The difference between a power-law random graph, here denoted by $P_{n,a}$ with
$a \geq 2$, and a random graph is its degree distribution. In $P_{n,a}$, the
probability that a randomly chosen vertex $v$ has degree $d > 0$ is
\begin{eqnarray*}
	 \Pr[d(v) = d] = \frac{d^{-a}}{\delta(a)}\mbox{,}
\end{eqnarray*}
where $\delta(a) = \sum_{k=1}^{n-1} k^{-a}$.

Note that $\delta(a) \tendingtoinfty{n} \zeta(a)$, where $\zeta(a)$ is the
Riemann zeta function, of which it is known that
$\zeta(a) \tendingtoinfty{a} 1$. For example, $\zeta(2) \approx 1.64$,
$\zeta(3) \approx 1.20$, and $\zeta(4) \approx 1.08$, which in turn are the
approximate limits of $\delta(2)$, $\delta(3)$, and $\delta(4)$, respectively,
when $n \to \infty$.

\begin{thm}
\label{EXPna}
If $G$ is a power-law random graph $P_{n, a}$, then
$\mathrm{E}[X_n] \geq \frac{n}{2 \delta(a)}$, for all $n \geq 1$.
\end{thm}

\begin{proof}
We have
\begin{eqnarray*}
	 \mathrm{E}[X_n] & = & \sum_{v \in V(G)} \Pr[Y_v = 1] \\
		  & = & n \Pr[Y_1 = 1] \\
		  & = & n \sum_{d = 1}^{n - 1} \Pr[Y_1 = 1 \st d(1) = d] \Pr[d(1) = d] \\ 
		  & = & n \sum_{d = 1}^{n - 1}  \frac{1}{2^d} \times \frac{d^{-a}}{\delta(a)} \\
          & \geq & \frac{n}{2 \delta(a)} \;{.}
\end{eqnarray*}
\end{proof}

\begin{thm}
\label{PXPna}
If $G$ is a power-law random graph $P_{n, a}$, then
$\Pr[X_n > 0] \geq 1 - \left( 1 - \frac{2 - \sqrt{2}}{2\delta(a)} \right)^{n}$,
for all $n \geq 1$.
\end{thm}

\begin{proof}
We have
\begin{eqnarray*}
	 \Pr[X_n = 0] & = & \Pr[\bigwedge_{v \in V(G)} Y_v = 0] \\
	& \leq & \Pr[\bigwedge_{v \in S} Y_v = 0],
\end{eqnarray*}
where $S$ denotes the subset of degree-$1$ vertices of $V(G)$. Then
\begin{eqnarray*}
	 \Pr[X_n = 0] & \leq & \sum_{s=0}^n \Pr[\bigwedge_{v \in S} Y_v = 0 \st |S|=s] \Pr[|S|=s] \;{.}	
\end{eqnarray*}

Let $I \subseteq S$ be a maximum independent set of $G[S]$. Therefore,
$|I| \geq \frac{|S|}{2}$. Consequently,
\begin{eqnarray*}
	 \Pr[\bigwedge_{v \in S} Y_v = 0 \st |S|=s] & \leq & \Pr[\bigwedge_{v \in I} Y_v = 0 \st |S|=s]  \leq \left( \frac{1}{2} \right)^{s/2} \;{.}
\end{eqnarray*}
On the other hand, since $\Pr[d(v) = 1] = 1 / \delta(a)$ for each $v \in V(G)$, 
\begin{eqnarray*}
	  \Pr[|S|=s] & = & \binom{n}{s} \left( \frac{1}{\delta(a)} \right)^{s} \left( 1 - \frac{1}{\delta(a)} \right)^{n - s} \;{,}
\end{eqnarray*}
and therefore,
\begin{eqnarray*}
	 \Pr[X_n = 0] & \leq & \sum_{s=0}^n \left( \frac{1}{2} \right)^{s/2} \binom{n}{s} \left( \frac{1}{\delta(a)} \right)^{s} \left( 1 - \frac{1}{\delta(a)} \right)^{n - s} \\
				    & = & \sum_{s=0}^n \binom{n}{s} \left( \frac{1}{\sqrt{2}\delta(a)} \right)^{s} \left( 1 - \frac{1}{\delta(a)} \right)^{n - s} \\
				    & = & \left( \frac{1}{\sqrt{2}\delta(a)} + 1 - \frac{1}{\delta(a)} \right)^{n} \\
				    & = & \left( 1 - \frac{2 - \sqrt{2}}{2\delta(a)} \right)^{n} \;{.}
\end{eqnarray*}

Consequently,
$\Pr[X_n > 0] = 1 - \Pr[X_n = 0] \geq 1 - \left( 1 - \frac{2 - \sqrt{2}}{2\delta(a)} \right)^{n} \;{.}$
\end{proof}

It follows from Theorems~\ref{EXPna} and~\ref{PXPna} that
$\mathrm{E}[X_n] \tendingtoinfty{n} \infty$ and
$\Pr[X_n > 0] \tendingtoinfty{n} 1$, respectively.

\section{Summary}

Table~\ref{tab:Resumo} summarizes our finds in Section $5$. They all refer to
the particular case in which each edge of $G$ corresponds to a resource. By
Theorem~\ref{CaracLivreDeadlock:Teorico}, therefore, any deadlock-free RGM must
guarantee the acyclicity of $G$ at all times. Our approach to accomplish this
has been probabilistic, of the Las Vegas type, in an attempt at feasibility even
in the case of large computations. What can be expected in the general case, as
well as questions of fairness, remain research topics.

\begin{table}[t]
	\caption{Summarization of $\Pr[X_n > 0]$ and $\mathrm{E}[X_n]$ for several classes of graphs.}
	\label{tab:Resumo}
	\begin{center}
	\begin{tabular}{|l|c|c|}
		\hline
		Class & $\Pr[X_n > 0]$ & $\mathrm{E}[X_n]$ \\
		\hline
		\hline
		$T_n$ & $1$ &  \begin{tabular}{c}
									$P_n$:  $\frac{n+2}{4}$ \\
									$S_n$: $\frac{n-1}{2} + \frac{1}{2^{n-1}}$
									\end{tabular}
									\\
		\hline
		$C_n$ & $1 - \frac{1}{2^{n-1}}$ & $\frac{n}{4}$ \\
		\hline
		$K_n$ & $\frac{n}{2^{n-1}}$ & $\frac{n}{2^{n-1}}$ \\
		\hline
		$B_{n, k}$ & $\geq 1 - \left( 1 - \frac{1}{2^k} \right)^{\left\lceil \frac{n}{k+1} \right\rceil }$ & $\geq \frac{1}{2^k} \left\lceil \frac{n}{k+1} \right\rceil$ \\
		\hline
		$G_{n, p}$ & $\tendingtoinfty{n} 0$ & $\approx \frac{n}{e^{p(n-1)/2}}$ \\
		\hline
		$P_{n, a}$ & $\geq 1 - \left( 1 - \frac{2 - \sqrt{2}}{2\delta(a)} \right)^{n}$ & $\geq \frac{n}{2 \delta(a)}$ \\
		\hline
	\end{tabular}
	\end{center}
\end{table}

\subsection*{Acknowledgments}
We acknowledge partial support from CNPq, CAPES, and a FAPERJ BBP grant.

\bibliographystyle{plain}
\bibliography{PrevencaoDeadlock}

\begin{thebibliography}{1}

\bibitem{LivroValmir}
V.~C. Barbosa.
\newblock {\em An Introduction to Distributed Algorithms}.
\newblock The MIT Press, Cambridge, MA, 1996.

\bibitem{Barbosa99thecombinatorics}
V.~C. Barbosa.
\newblock The combinatorics of resource sharing.
\newblock In R.~Corr\^{e}a, I.~Dutra, M.~Fiallos, and F.~Gomes, editors, {\em
  Models for Parallel and Distributed Computation: Theory, Algorithmic
  Techniques and Applications}. Kluwer Academic Publishers, Dordrecht, The
  Netherlands, 2002.

\bibitem{BrachaToueg}
G.~Bracha and S.~Toueg.
\newblock Distributed deadlock detection.
\newblock {\em Distributed Computing}, 2:127--138, 1987.

\bibitem{BrzezinskiHelaryRaynalSinghal}
J.~Brzezinski, J.-M. H\'{e}lary, M.~Raynal, and M.~Singhal.
\newblock Deadlock models and a general algorithm for distributed deadlock
  detection.
\newblock {\em Journal of Parallel and Distributed Computing}, 31:112--125,
  1995.

\bibitem{GareyJohnson}
M.~R. Garey and D.~S. Johnson.
\newblock {\em Computers and Intractability: A Guide to the Theory of
  NP-Completeness}.
\newblock W. H. Freeman, New York, NY, 1979.

\bibitem{KshemkalyaniSinghal}
A.~D. Kshemkalyani and M.~Singhal.
\newblock Efficient detection and resolution of generalized distributed
  deadlocks.
\newblock {\em IEEE Transactions on Software Engineering}, 20:43--54, 1994.

\bibitem{MisraChandy}
J.~Misra and K.~M. Chandy.
\newblock A distributed graph algorithm: knot detection.
\newblock {\em ACM Transactions on Programming Languages and Systems},
  4:678--686, 1982.

\bibitem{RyangPark}
D.-S. Ryang and K.~H. Park.
\newblock A two-level distributed detection algorithm of {AND/OR} deadlocks.
\newblock {\em Journal of Parallel and Distributed Computing}, 28:149--161,
  1995.

\end{thebibliography}

\end{document}